%% file: main.tex
\documentclass[sigconf]{acmart}
\usepackage{enumitem}
\usepackage{booktabs}
\usepackage[table]{xcolor}
\usepackage{multirow}
\newcommand{\baselineSep}{%
  \arrayrulecolor{black!30}%
  \specialrule{0.3pt}{2pt}{2pt}%
  \arrayrulecolor{black}%
}

\AtBeginDocument{%
  }

\setcopyright{acmlicensed}
\copyrightyear{2018}
\acmYear{2018}
\acmDOI{XXXXXXX.XXXXXXX}
\acmConference[Conference acronym 'XX]{Make sure to enter the correct
  conference title from your rights confirmation email}{June 03--05,
  2018}{Woodstock, NY}
\acmISBN{978-1-4503-XXXX-X/2018/06}

\begin{document}

\title{From Item Scores to Semantic Tokens: Test-Time Collaborative Adaptation for Generative Recommendation}

\author{Xiongfeng Shan}
\authornote{Both authors contributed equally to this research.}
\affiliation{%
 \institution{Shenzhen Technology University}
 \city{Shenzhen}
 \country{China}}

\author{Xing Tang}
\authornotemark[1]
\authornote{Corresponding authors.}
\affiliation{%
  \institution{Shenzhen Technology University}
  \city{Shenzhen}
  \country{China}}

\author{Hanchong Chen}
\affiliation{%
  \institution{Shenzhen Technology University}
  \city{Shenzhen}
  \country{China}}

\author{Fuyuan Lyu}
\affiliation{%
  \institution{McGill University}
  \city{Montreal}
  \country{Canada}}
%\email{jsmith@affiliation.org}

\author{Lingjie Li}
\affiliation{%
  \institution{Shenzhen Technology University}
  \city{Shenzhen}
  \country{China}}
%\email{jpkumquat@consortium.net}

\author{Jingyan Jiang}
\affiliation{%
  \institution{Shenzhen Technology University}
  \city{Shenzhen}
  \country{China}}

\author{Dugang Liu}
\affiliation{%
  \institution{Shenzhen University}
  \city{Shenzhen}
  \country{China}}

\author{Xiuqiang He}
\authornotemark[2]
\affiliation{%
  \institution{Shenzhen Technology University}
  \city{Shenzhen}
  \country{China}}

\renewcommand{\shortauthors}{Trovato et al.}
\input{section/abstract}



\maketitle
\input{section/introduction}
\input{section/rw}
\input{section/method}
\input{section/experiment}
\input{section/conclusion}

\clearpage

\bibliographystyle{ACM-Reference-Format}
\bibliography{ref}

%\clearpage
\appendix
\input{section/appendix}
\end{document}

%% file: section/abstract.tex
\begin{abstract}
Generative recommenders using semantic identifiers predict items through successive token decisions, but beam search may irreversibly prune promising item paths.
This limitation is particularly consequential when a few recent interactions signal an emerging interest while earlier behavior favors established preferences.
Yet an independently trained discriminative recommender may rank the same items highly from the same interaction history.
Such complementary scores are available at test time, offering a way to adapt a frozen generator without retraining either model.
The challenge is to use these scores before pruning: combining completed recommendation lists cannot alter earlier beam decisions, while scores over complete items cannot directly inform choices among semantic prefixes shared by multiple items.
We propose \textbf{T}est-Time \textbf{C}ollaborative \textbf{B}eam \textbf{S}earch \textbf{A}daptation (\textbf{TCBSA}), a training-free method that incorporates collaborative information into semantic decoding.
TCBSA transforms discriminative score margins into potentials over shared semantic edges and uses a confidence gate to regulate their contributions to cumulative path scores before beam selection.
We establish a sufficient condition for target prefix survival under global beam competition and bound score perturbations under low discriminative confidence.
Full-catalog experiments show consistent improvements over strong baselines across datasets and different generative and discriminative recommenders.
Analyses of emerging interests and decoding behavior demonstrate gains in recommendation accuracy and target path survival, while targeted interventions support the contribution of shared prefix corrections.
TCBSA also offers a favorable balance between recommendation accuracy and inference latency compared with simply widening the native beam.
%Code is available at \url{https://anonymous.4open.science/r/TCBSA-35DD/}.
\end{abstract}

%% file: section/introduction.tex
\section{Introduction}

\begin{figure}[t]
    \centering
    \includegraphics[width=0.95\linewidth]{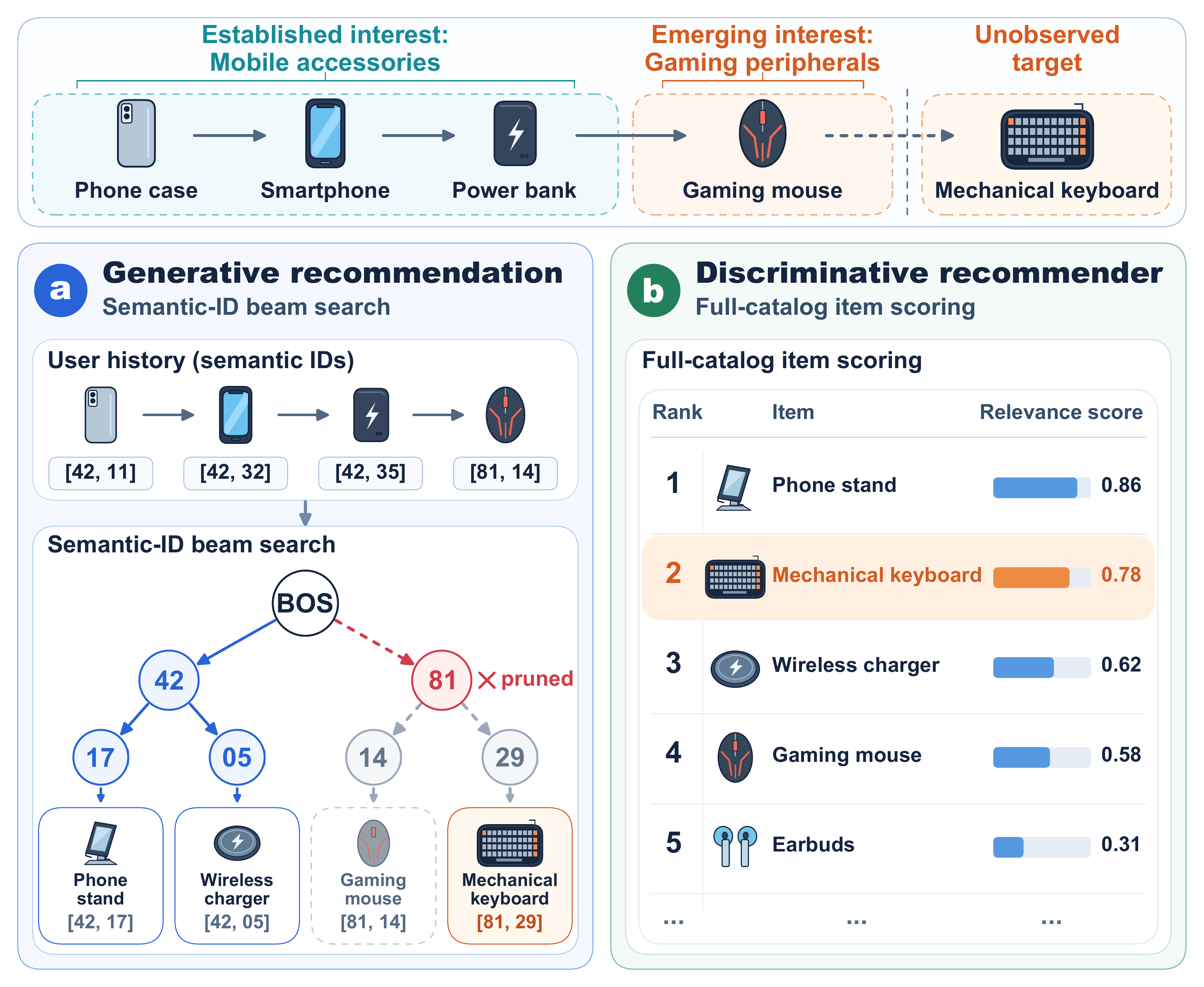}
    \caption{
    A test-time disconnect between item-level ranking and semantic decoding.
    (a) The frozen generator prunes prefix 81, shared by the recently interacted gaming mouse and the unobserved mechanical-keyboard target, rendering the latter unreachable.
    (b) The frozen discriminative recommender nevertheless ranks the target second over the full catalog.
    A simplified two-level semantic tree is shown.
    }
    \label{fig:motivation}
    \Description{
    A user history transitions from mobile accessories to a gaming mouse, with an unobserved mechanical keyboard as the target.
    Panel A shows semantic-ID beam search pruning token 81, excluding the target path.
    Panel B shows full-catalog item scores ranking the mechanical keyboard second.
    }
\end{figure}

Generative recommendation has emerged as a promising paradigm for next-item prediction through autoregressive generation of item identifiers~\cite{tiger,onerec,onerec1}.
In semantic-ID-based approaches, shared semantic tokens capture commonalities across items and organize next-item prediction within a structured generative space~\cite{tiger,LC-Rec}.
At test time, however, beam search retains only the highest-scoring partial identifiers at each depth, making early pruning irreversible: once a prefix is discarded, every item beneath it becomes unreachable~\cite{apao,bear}.
This limitation is particularly pronounced when an emerging interest is reflected by only a few recent interactions, while a longer and more coherent portion of the history favors an established preference.
As illustrated in Figure~\ref{fig:motivation}(a), the frozen generative recommender may preserve prefixes aligned with the established interest while prematurely pruning the branch containing the emerging-interest target.
A potentially relevant item can thus be lost before its complete semantic identifier is ever generated.

The same user history can nonetheless yield a different ranking under discriminative recommendation, which scores complete catalog items directly~\cite{gru4rec,sasrec,bert4rec,caser}.
Unlike generative beam search, discriminative scoring does not rely on successive decisions over semantic prefixes.
As illustrated in Figure~\ref{fig:motivation}(b), the frozen discriminative recommender ranks the emerging-interest target second over the full catalog, even though the generative decoder has pruned its semantic path.
Recent comparative findings further highlight their complementary strengths: discriminative models better capture observed item transitions, whereas generative models perform better on cases requiring generalization~\cite{memgr}.
Crucially, an independently trained discriminative recommender can provide history-specific item scores at test time without retraining the generator.
This creates an opportunity for \textit{test-time collaboration}: using complementary item-level information to adapt semantic-prefix selection while decoding is still underway.

Existing approaches incorporate complementary signals and improve generative recommendation at different stages.
Training-time methods incorporate collaborative signals or behavioral priors~\cite{human,TCA4rec}, while prefix-aware objectives improve beam-search behavior~\cite{apao,bear}; both require modifying generator training.
Post-hoc fusion combines completed discriminative and generative recommendation lists~\cite{memgr}, but cannot alter prefixes already pruned by the generator.
Other inference-time methods employ retrieval-based verification~\cite{specgr} or reward- and value-guided search~\cite{promise,vstar}.
In sequential recommendation, test-time adaptation has also been explored through parameter updates~\cite{dt3or,t2arec} and retrieval-based item-level refinement~\cite{read}.
Bringing collaborative scores from an independently trained discriminative recommender into the intermediate decisions of a frozen semantic-ID generator, however, presents a distinct technical difficulty.
The two operate in different prediction spaces: the former scores complete items, whereas the latter selects prefix-conditioned token extensions that may be shared by multiple item paths.
An item-level score therefore cannot be applied directly to a token-level decision.
Under irreversible beam pruning, this mismatch raises a fundamental question: how can item-level collaborative scores be brought into the decoding decisions of a frozen generative recommender at test time, without updating either model?

To address this question, we propose \textbf{T}est-Time \textbf{C}ollaborative \textbf{B}eam \textbf{S}earch \textbf{A}daptation (\textbf{TCBSA}), a training-free method that adapts frozen generative decoding using item-level collaborative scores.
Rather than combining completed recommendation lists, TCBSA transforms collaborative item scores into corrections to the semantic paths under exploration, allowing them to influence beam selection before pruning.
For each user history, a frozen discriminative recommender scores the full catalog, and TCBSA computes score margins relative to the discriminative Top-$R$ cutoff, making them invariant to additive shifts in item scores.
Using a precomputed item-to-path mapping, these margins are aggregated over shared prefix--token edges to construct history-specific collaborative potentials.
Because multiple item paths can share semantic edges, the resulting potentials can also support items outside the discriminative Top-$R$ set.
To regulate their influence, a confidence gate adjusts the intervention strength according to the maximum softmax-normalized discriminative score for the current history.
During a single decoding pass, TCBSA incorporates these confidence-tempered potentials into cumulative path scores before each beam-pruning step, adapting the search while keeping both models frozen.

Our main contributions are summarized as follows:
\begin{itemize}[leftmargin=*]

\item We propose TCBSA, a test-time adaptation method that incorporates collaborative scores from an independently trained discriminative recommender into the beam decisions of a frozen generative recommender. It projects item-score margins onto shared semantic edges to intervene before pruning, without model updates.

\item We characterize TCBSA as a history-specific transformation of cumulative generative path scores and derive a sufficient condition for target-prefix survival under global beam competition. We further analyze the score perturbations induced by confidence-tempered adaptation and characterize its inference complexity.

\item Extensive full-catalog experiments demonstrate improvements across four datasets and three generative--discriminative model pairings. Analyses of emerging interests, decoding behavior, shared-prefix attribution, and beam widths further establish the effectiveness and accuracy--latency trade-offs of test-time collaborative adaptation.

\end{itemize}

%% file: section/rw.tex
\section{Related Work}

\paragraph{Generative Recommendation.}
Semantic-ID-based generative recommendation represents items as sequences of discrete tokens and predicts the next item through autoregressive decoding~\cite{rq-vae,tiger}.
Subsequent work improves identifier construction and sequence modeling through collaborative semantics and textual identifiers~\cite{LC-Rec,IDGenRec}.
Recent frameworks such as MiniOneRec~\cite{minionerec} further advance generative recommendation through scalable model training and preference alignment.
At training time, collaborative and behavioral information is incorporated through learned identifiers, dual-stream modeling, and behavioral prior heads~\cite{letter,eager,human}.
TCA4Rec projects collaborative logits into token distributions for soft supervision during generator training~\cite{TCA4rec}.
APAO and BEAR~\cite{apao,bear} address the training--inference gap through prefix-aware optimization and beam-aware regularization, respectively.
Other architectures combine semantic generation with dense representations or item-level scoring, as exemplified by COBRA~\cite{cobra} and Gryphon~\cite{gryphon}.
At inference, SpecGR~\cite{specgr} verifies retrieved candidates with the generator, while LIGER~\cite{liger} combines generative retrieval with dense candidate reranking.
Recent methods also employ trained process reward models or value-aligned components to influence prefix search~\cite{promise,vstar,univa}.
MemGen~\cite{memgr} instead adaptively combines completed discriminative and generative recommendation lists.
TCBSA operates at a different interface: it projects retrieval-relative item-score margins onto shared prefix--token edges and adapts the beam scores of a frozen generator before pruning.

\paragraph{Guided Decoding and Test-Time Adaptation.}
Beyond recommendation, decoding-time control reweights a frozen generator using auxiliary attribute scores~\cite{fudge}, constrains beam search to valid outputs~\cite{constrained_bs,autoregressive}, or uses a stronger model to select among a weaker generator's top-$K$ token candidates~\cite{select}.
More broadly, test-time adaptation adjusts model predictions using information available during inference.
Classical approaches often update model parameters or normalization statistics from unlabeled test inputs~\cite{ttt,tent}, while retrieval-augmented methods incorporate external information into instance-level predictions~\cite{ratta,ttraa}.
In recommendation, DT3OR and T2ARec perform test-time parameter adaptation for distribution or interest shifts~\cite{dt3or,t2arec}, whereas ReAd~\cite{read} retrieves collaborative items and refines item-level predictions using an augmented representation.
TCBSA performs history-specific test-time adaptation within semantic decoding, incorporating collaborative scores from an independently trained discriminative recommender to modify beam selection while keeping both models frozen.

%% file: section/method.tex
\section{Problem Formulation}
\label{sec:problem}

Let $\mathcal{I}$ denote the item catalog and $\mathbf{x}_u=(i_{u,1},\ldots,i_{u,T})$ denote the interaction history of user $u$, where $T$ is the history length and $i_{u,t}\in\mathcal{I}$ is the item interacted with at position $t$.
Given $\mathbf{x}_u$, next-item recommendation aims to predict the next item $i_u^{+}\in\mathcal{I}$ that user $u$ will interact with.
We consider two recommendation paradigms that operate at different prediction granularities during inference.

\paragraph{Generative recommendation.}
A generative recommender $G$ represents each item $i$ by a length-$L$ semantic identifier $c(i)$.
We write $c(i)=(c_1(i),\ldots,c_L(i))$, where $c_\ell(i)\in\mathcal{V}_\ell$ is the token at depth $\ell$ and $\mathcal{V}_\ell$ is the corresponding token vocabulary.
The identifier is modeled autoregressively as:
\begin{equation}
P_G(c(i)\mid\mathbf{x}_u)
=
\prod_{\ell=1}^{L}
P_G\!\left(c_\ell(i)\mid c_{<\ell}(i),\mathbf{x}_u\right),
\label{eq:generative_factorization}
\end{equation}
where $P_G$ denotes the conditional distribution induced by $G$, and $c_{<\ell}(i)=(c_1(i),\ldots,c_{\ell-1}(i))$ is the prefix before depth $\ell$.
At test time, $G$ expands semantic prefixes token by token and retains the globally highest-scoring $B$ partial identifiers at each depth, where $B$ is the beam width~\cite{NCI}.
Once a target prefix is pruned, its descendant items become unreachable~\cite{apao,bear}.

\paragraph{Discriminative recommendation.}
A discriminative recommender $D$ instead assigns an item-level score $s_D(i\mid\mathbf{x}_u)$ to every $i\in\mathcal{I}$.
At test time, it retrieves:
\begin{equation}
\mathcal{R}_u^R
=
\underset{i\in\mathcal{I}}
{\operatorname{Top}\text{-}R}
\;s_D(i\mid\mathbf{x}_u),
\label{eq:discriminative_retrieval}
\end{equation}
where $R$ is the retrieval width and $\mathcal{R}_u^R$ contains the $R$ highest-scoring items for user $u$.

\paragraph{Test-time collaborative adaptation.}
The two paradigms operate in different prediction spaces: $D$ scores complete items, whereas $G$ makes prefix-conditioned token decisions.
Given independently trained and frozen $G$ and $D$, we seek a test-time adaptation mechanism $\mathcal{F}_{\mathrm{test}}$ that incorporates collaborative scores from $D$ into the intermediate decoding decisions of $G$ without updating either model.
Let $\mathbf{s}_D(\mathbf{x}_u)=\{s_D(i\mid\mathbf{x}_u)\}_{i\in\mathcal{I}}$ denote the item-level scores produced by $D$.
Our objective is to obtain an adapted recommendation list:
\begin{equation}
\widehat{\mathcal{L}}_u^K
=
\mathcal{F}_{\mathrm{test}}
\left(
G,\mathbf{x}_u,\mathbf{s}_D(\mathbf{x}_u)
\right),
\label{eq:test_time_adaptation}
\end{equation}
where $\widehat{\mathcal{L}}_u^K$ contains the top-$K$ recommended items.
The central challenge is to make item-level collaborative scores actionable during semantic decoding, before potentially useful paths are irreversibly pruned and without access to the ground-truth next item.

\section{Methodology}
\label{sec:method}

To bridge item-level collaborative ranking and token-level generative decoding, we propose \textbf{T}est-Time \textbf{C}ollaborative \textbf{B}eam \textbf{S}earch \textbf{A}daptation (\textbf{TCBSA}).
TCBSA implements $\mathcal{F}_{\mathrm{test}}$ in Equation~\eqref{eq:test_time_adaptation} by adjusting the frozen generator's beam scores using discriminative item scores.
As illustrated in Figure~\ref{fig:method_overview}, TCBSA comprises an offline preparation step and two test-time stages.
Offline, fixed semantic identifiers are used to construct an item-to-path lookup that records each item's prefix--token edges.
At test time, the frozen discriminative recommender $D$ first scores the full item catalog for the current history $\mathbf{x}_u$.
TCBSA computes score margins relative to the Top-$R$ cutoff and aggregates them over shared semantic edges to construct collaborative potentials, whose influence is controlled by a confidence gate.
The frozen generator $G$ then performs a single decoding pass, during which these potentials adjust cumulative path scores before each beam-pruning step.

\begin{figure*}[t]
    \centering
    \includegraphics[width=0.98\textwidth]{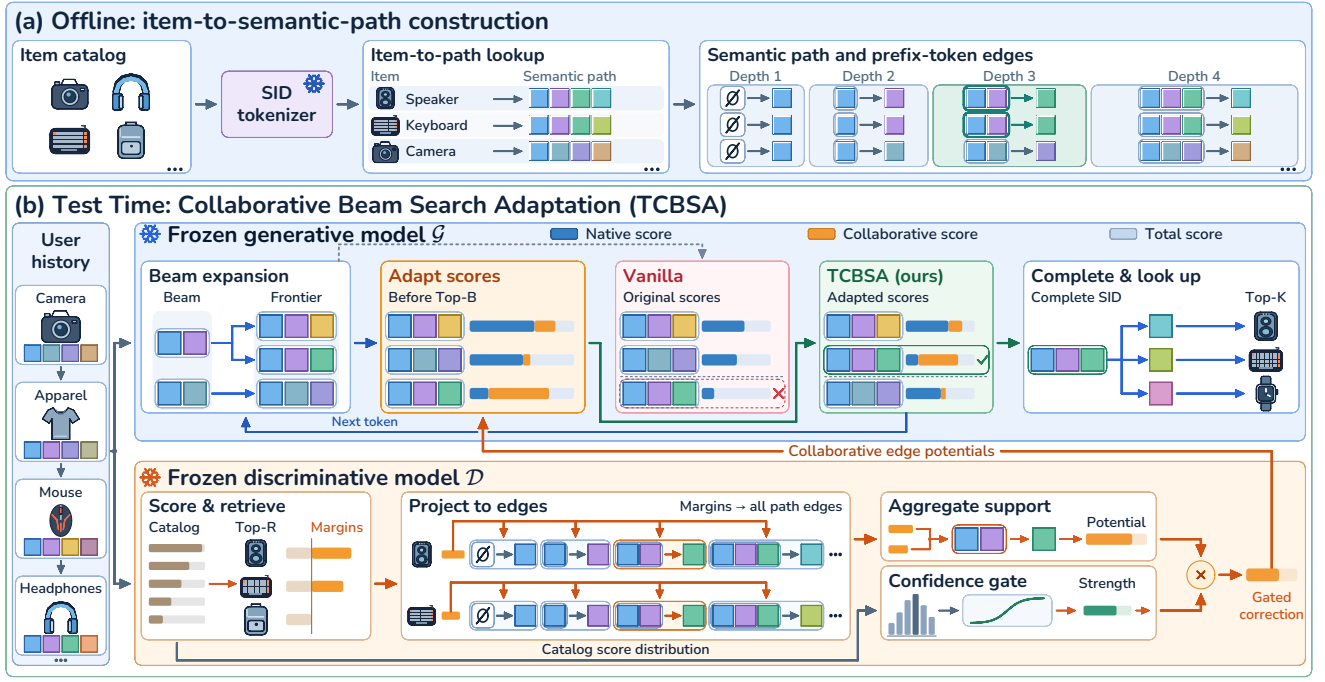}
    \caption{
    Overview of TCBSA.
    (a) Offline construction of an item-to-path lookup using fixed semantic identifiers.
    (b) At test time, scores from a frozen discriminative recommender are transformed into shared semantic-edge potentials and scaled by a confidence gate to adapt the frozen generator's beam scores before pruning.
    Shared edges can also support items beyond the discriminative Top-$R$ set.
    }
    \label{fig:method_overview}
\end{figure*}

\subsection{Cross-Space Collaborative Score Transformation}
\label{sec:collaborative_evidence}

Item-level scores from $D$ cannot directly determine how a semantic prefix should be extended.
TCBSA bridges this mismatch by constructing score margins relative to the discriminative Top-$R$ cutoff and projecting them onto prefix--token edges along the corresponding semantic paths.

\paragraph{Top-$R$ score margins.}

The discriminative Top-$R$ set identifies high-scoring items, but their absolute scores do not directly specify the strength of their contributions to semantic decoding.
Because discriminative scores have no intrinsic zero point, using them directly would make the resulting corrections sensitive to arbitrary additive shifts.
Let $s_D^{(R)}(\mathbf{x}_u)$ denote the $R$-th largest discriminative score for history $\mathbf{x}_u$.
We define the item-score margin relative to this cutoff as:
\begin{equation}
e_i(\mathbf{x}_u)
=
\left[
s_D(i\mid\mathbf{x}_u)
-
s_D^{(R)}(\mathbf{x}_u)
\right]_+,
\label{eq:item_evidence}
\end{equation}
where $[z]_+=\max(z,0)$ denotes the positive-part operator.
The margin measures how far an item's score lies above the Top-$R$ cutoff.
Items at or below the cutoff receive zero margins, and the resulting margins are invariant to additive shifts of the discriminative scores.

\paragraph{Projecting item scores onto semantic edges.}

The score margins are defined over complete items, whereas the generator makes successive decisions over prefix-conditioned token extensions.
We therefore propagate each Top-$R$ item's margin along its semantic path.
For an item $i$, we define its set of prefix--token edges as:
\begin{equation}
\mathcal{E}(i)
=
\left\{
\big(c_{<\ell}(i),c_\ell(i)\big)
\right\}_{\ell=1}^{L},
\label{eq:item_path_edges}
\end{equation}
where $\mathcal{E}(i)$ contains the $L$ edges traversed by the semantic identifier $c(i)$.
Because semantic identifiers are fixed, $\mathcal{E}(i)$ can be precomputed for every item and reused across user histories.

Let $p$ denote a current semantic prefix and $t$ a candidate token extending that prefix.
We aggregate the margins of Top-$R$ items whose semantic paths contain the edge $(p,t)$:
\begin{equation}
A_{\mathbf{x}_u}(p,t)
=
\operatorname{Agg}
\left(
\left\{
e_i(\mathbf{x}_u)
\,\middle|\,
i\in\mathcal{R}_u^R,\,
(p,t)\in\mathcal{E}(i)
\right\}
\right),
\label{eq:prefix_edge_potential}
\end{equation}
where $A_{\mathbf{x}_u}(p,t)$ is the collaborative potential assigned to edge $(p,t)$, and $\operatorname{Agg}$ combines the item-score margins associated with that edge.
For top-$M$ mean aggregation, let $v_{(1)}\geq\cdots\geq v_{(n)}$ denote the $n$ margin values and $\bar{M}=\min(M,n)$, where $M$ is the aggregation budget.
The potential is computed as $\bar{M}^{-1}\sum_{j=1}^{\bar{M}}v_{(j)}$; max aggregation instead returns $v_{(1)}$.
We set $A_{\mathbf{x}_u}(p,t)=0$ when no retrieved item supports the edge.

By aggregating item-score margins over shared semantic edges, TCBSA constructs prefix-conditioned collaborative potentials that directly influence token-level decoding decisions.
An item outside $\mathcal{R}_u^R$ can also benefit when its semantic path shares edges with Top-$R$ items.
Such corrections apply only to the shared edges; the remaining path must still be completed by the generative model.

\subsection{Confidence-Tempered Beam Adaptation}
\label{sec:beam_steering}

The collaborative potentials capture support for semantic edges but do not account for the confidence of the discriminative predictions.
Applying these corrections uniformly across histories may distort native beam scores when the discriminative predictions are uncertain.
TCBSA therefore adjusts the intervention strength according to the confidence of $D$, allowing history-specific collaborative adaptation during decoding.

\paragraph{Confidence-tempered intervention.}
We use the maximum softmax probability of the discriminative scores as a lightweight confidence proxy.
Specifically, we define the normalized scores and confidence as:
\begin{equation}
\begin{aligned}
q_D(i\mid\mathbf{x}_u)
&=
\frac{\exp s_D(i\mid\mathbf{x}_u)}
{\sum_{j\in\mathcal{I}}\exp s_D(j\mid\mathbf{x}_u)}, \\
m(\mathbf{x}_u)
&=
\max_{i\in\mathcal{I}}q_D(i\mid\mathbf{x}_u),
\end{aligned}
\label{eq:discriminative_confidence}
\end{equation}
where $q_D(i\mid\mathbf{x}_u)$ is the softmax-normalized discriminative score and $m(\mathbf{x}_u)$ measures its maximum concentration over the full catalog.
We use $m(\mathbf{x}_u)$ as a relative confidence proxy rather than a calibrated probability of correctness.

The effective intervention strength is defined as:
\begin{equation}
\eta(\mathbf{x}_u)
=
\lambda
\sigma\!\left(
k\big(m(\mathbf{x}_u)-\tau\big)
\right),
\label{eq:confidence_gate}
\end{equation}
where $\sigma(a)=1/(1+\exp(-a))$ is the sigmoid function, $\lambda>0$ is the maximum intervention scale, $k>0$ controls the sharpness of the gate, and $\tau$ is the confidence threshold.
The strength $\eta(\mathbf{x}_u)$ increases with discriminative confidence and is attenuated when confidence falls below $\tau$.

\paragraph{Pre-pruning beam adaptation.}
Let $\pi_0(t\mid p,\mathbf{x}_u)=P_G(t\mid p,\mathbf{x}_u)$ denote the native next-token distribution of $G$ for candidate token $t$ under prefix $p$.
TCBSA incorporates the confidence-tempered collaborative potential into the cumulative path score before beam selection:
\begin{equation}
\widetilde{S}_u(p\oplus t)
=
\widetilde{S}_u(p)
+
\log \pi_0(t\mid p,\mathbf{x}_u)
+
\eta(\mathbf{x}_u)A_{\mathbf{x}_u}(p,t),
\label{eq:guided_path_score}
\end{equation}
where $\widetilde{S}_u(p)$ is the accumulated adapted score of prefix $p$, $p\oplus t$ denotes its extension by token $t$, and $\widetilde{S}_u(\epsilon)=0$ for the empty prefix $\epsilon$.
At each decoding depth, TCBSA retains the globally highest-scoring $B$ extensions according to Equation~\eqref{eq:guided_path_score}, repeating this process until $L$ tokens have been generated.

For a partial semantic path $h$, let $S_0(h\mid\mathbf{x}_u)$ denote its native cumulative log-score under $G$.
Unrolling the recursive update yields:
\begin{equation}
\begin{aligned}
\widetilde{S}_u(h)
&=
S_0(h\mid\mathbf{x}_u)+\Delta_u(h),\\
\Delta_u(h)
&=
\eta(\mathbf{x}_u)
\sum_{j=1}^{|h|}
A_{\mathbf{x}_u}(h_{<j},h_j).
\end{aligned}
\label{eq:collaborative_score_correction}
\end{equation}
This decomposition characterizes TCBSA as a history-specific test-time adaptation of cumulative path scores, while the generator's native token probabilities and model parameters remain unchanged.

An unsupported edge has $A_{\mathbf{x}_u}(p,t)=0$ and receives no additional local correction, although its preceding prefix may have accumulated collaborative support.
By applying these corrections before global beam pruning, TCBSA can preserve promising semantic paths that would otherwise become unreachable.
After $L$ steps, valid identifiers are mapped to items and ranked by their adapted scores, yielding $\widehat{\mathcal{L}}_u^K$.

Appendix~\ref{app:theory_proofs} establishes a sufficient condition for target-prefix survival under global beam competition, showing when the collaborative advantage, scaled by the confidence gate, overcomes the native score deficit.
It also provides a stability bound on low-confidence score perturbations and analyzes inference complexity.

%% file: section/experiment.tex
\section{Experiments}
\label{sec:experiments}

We conduct experiments on four public datasets to address the following research questions:

\begin{itemize}[leftmargin=*]
\item \textbf{RQ1:}
How does TCBSA compare with existing methods, and does it generalize across generative backbones and discriminative recommenders?

\item \textbf{RQ2:}
How effectively does TCBSA adapt to emerging interests supported by only one or two recent interactions?

\item \textbf{RQ3:}
How does TCBSA reshape semantic beam search, and can it generate correct recommendations beyond the discriminative Top-$R$ set?

\item \textbf{RQ4:}
How do confidence tempering, semantic-edge aggregation, and shared-prefix corrections contribute to recommendation outcomes?

\item \textbf{RQ5:}
How does TCBSA's accuracy--latency trade-off change with beam width?
\end{itemize}

\input{tables/main_table.tex}

\subsection{Experimental Setup}
\label{sec:experimental_setup}

\paragraph{Datasets.}
We use three 5-core domains from Amazon Reviews 2023~\cite{AmazonReviews2023}: Industrial and Scientific (\textit{Science}), Musical Instruments (\textit{Music}), and Office Products (\textit{Office}), alongside Yelp from Kaggle~\cite{YelpKaggleV2}, preprocessed following~\cite{S3Rec}.
The user/item counts are 50,985/25,848 (Science), 58,062/24,701 (Music), 225,726/77,929 (Office), and 30,431/20,033 (Yelp).
Following~\cite{tiger,sasrec,bert4rec}, we chronologically order interactions and adopt leave-one-out evaluation, reserving the last interaction for testing, the penultimate for validation, and the remainder for training.
Each input retains at most the 20 most recent preceding interactions~\cite{tiger,LC-Rec}.

\paragraph{Evaluation protocol.}
We evaluate against the full item catalog without candidate sampling~\cite{etegr,sampledmetrics}.
Previously interacted items remain eligible for all methods, and padding is masked.
We report Recall@$K$ and NDCG@$K$ for $K\in\{5,10\}$, using validation NDCG@10 for model selection and hyperparameter tuning.
Statistical significance is assessed via user-level paired bootstrap (10,000 resamples, $p<0.05$), comparing each TCBSA variant with its unadapted generative backbone and the strongest baseline for each dataset--metric pair.

\paragraph{Baselines.}
We compare TCBSA with four categories of methods.
For discriminative recommendation, we include GRU4Rec~\cite{gru4rec}, SASRec~\cite{sasrec}, and BERT4Rec~\cite{bert4rec}.
The generative baselines include the semantic-ID models TIGER~\cite{tiger} and MiniOneRec~\cite{minionerec}, alongside HSTU~\cite{HSTU}, a strong sequential generative architecture operating on raw item IDs.
Training-time enhancements include Human-Prior~\cite{human}, which incorporates behavioral priors through auxiliary prediction heads, and APAO-Pairwise~\cite{apao}, which employs prefix-aware optimization.
Finally, MemGen-Fixed and MemGen-Adaptive~\cite{memgr} serve as post-hoc fusion baselines.
They combine completed TIGER and SASRec recommendation lists using validation-selected fixed and confidence-dependent mixing weights, respectively.

\paragraph{Implementation details.}
For TIGER and SASRec, we follow the MemGen-GR implementation~\cite{memgr}.
GRU4Rec and BERT4Rec are implemented with RecBole~\cite{recbole}, while HSTU and APAO-Pairwise follow their official implementations, with parameter budgets matched to TIGER.
MiniOneRec undergoes supervised fine-tuning followed by reinforcement learning, with its final RL checkpoint used for evaluation.
We encode item textual metadata with Sentence-T5~\cite{sentence} and apply a three-level RQ-K-means quantizer~\cite{onerec}.
A fourth collision-resolution token is appended to obtain unique length-$L=4$ semantic identifiers.
For a controlled cross-backbone comparison, all semantic-ID-based models share the same item-to-semantic-ID mapping.
We evaluate TCBSA with three pairs of frozen generative and discriminative recommenders: TIGER--SASRec, TIGER--GRU4Rec, and MiniOneRec--SASRec.
The default beam width is $B=20$, and a shared retrieval width $R=500$ is chosen by macro-averaged validation NDCG@10.
The remaining adaptation hyperparameters are tuned on TIGER--SASRec for each dataset and reused unchanged for the other two pairings.
Further implementation details are provided in Appendix~\ref{app:implementation_details}.

\subsection{Overall Effectiveness and Generality (RQ1)}
\label{sec:overall_results}

Across the four datasets, TCBSA variants achieve the best results on all 16 dataset--metric pairs in Table~\ref{tab:overall_performance}, with significant gains over the strongest baselines in 14 comparisons.
The advantage is particularly pronounced on Yelp, where MiniOneRec + TCBSA(SASRec) improves over the strongest baselines by 8.4\%--12.4\%.
In the default TIGER--SASRec pairing, TCBSA also outperforms APAO-Pairwise, Human-Prior, and both post-hoc fusion baselines across all datasets and metrics.
Unlike completed-list fusion, which cannot alter the generator's earlier pruning decisions, TCBSA incorporates collaborative information while semantic paths remain accessible, enabling direct intervention in beam selection.

TCBSA's effectiveness extends across different generative backbones and discriminative recommenders.
Using SASRec as the discriminative model, TCBSA improves both TIGER and MiniOneRec on all 16 dataset--metric pairs, yielding NDCG@10 gains of 4.9\%--30.4\% and 2.3\%--9.6\%, respectively.
Replacing SASRec with GRU4Rec also improves TIGER across all metrics, with NDCG@10 gains of 3.2\%--14.9\%.
Importantly, all three configurations retain their frozen checkpoints and reuse dataset-specific adaptation hyperparameters selected solely on TIGER--SASRec, without pairing-specific tuning.

\input{tables/emerging_interest.tex}

\begin{figure*}[t]
\centering

\begin{minipage}[t]{0.3\textwidth}
\vspace{0pt}
\centering
\includegraphics[width=\linewidth]{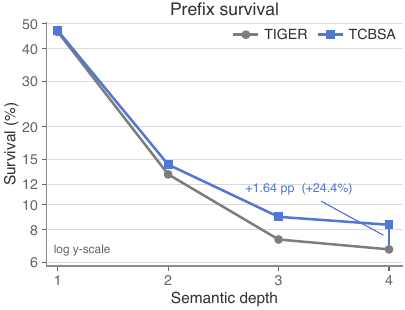}
\par\smallskip
\textbf{(a)}
\end{minipage}
\hfill
\begin{minipage}[t]{0.3\textwidth}
\vspace{0pt}
\centering
\includegraphics[width=\linewidth]{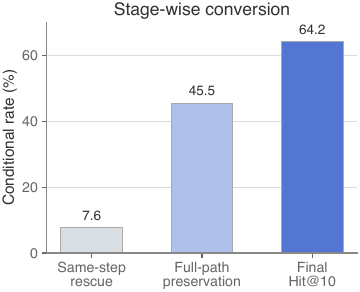}
\par\smallskip
\textbf{(b)}
\end{minipage}
\hfill
\begin{minipage}[t]{0.314\textwidth}
\vspace{0pt}
\centering
\includegraphics[width=\linewidth]{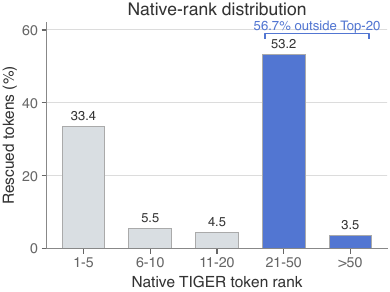}
\par\smallskip
\textbf{(c)}
\end{minipage}

\caption{
Beam-search behavior on Yelp with $B=20$.
(a) Target-prefix survival across semantic depths.
(b) Conditional progression from same-step rescue to full-path preservation and final Hit@10.
(c) Native TIGER token-rank distribution at the first target-pruning depth, conditioned on the ground-truth preceding prefix.
}
\label{fig:search_mechanism}
\Description{
Three panels showing target-prefix survival, stage-wise rescue propagation, and the native token-rank distribution of rescued targets.
}
\end{figure*}

\begin{figure}[t]
    \centering
    \makebox[\columnwidth][c]{%
        \hspace*{-0.025\columnwidth}%
        \includegraphics[
            width=0.70\columnwidth
        ]{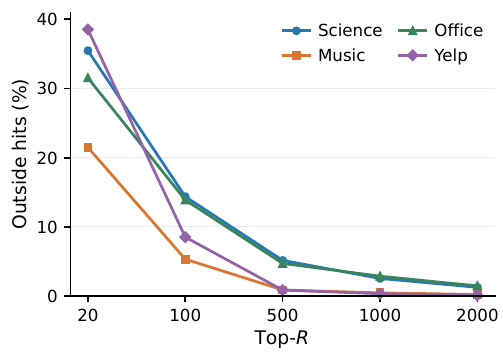}%
        \hspace*{0.025\columnwidth}%
    }
    \caption{
    Correct recommendations beyond the discriminative Top-$R$ set.
    Each point shows the proportion of correct TCBSA Top-10 hits outside the corresponding SASRec Top-$R$ set.
    Top-$R$ values are displayed at equally spaced positions.
    }
    \label{fig:outside_hits}
    \Description{
    Outside-retrieval correct-hit percentages across five
    retrieval widths on four datasets.
    }
\end{figure}

\subsection{Adaptation to Emerging Interests (RQ2)}
\label{sec:emerging_interest}

We evaluate test-time adaptation to emerging interests on Science and Yelp, focusing on histories where only a few recent interactions align with the target, following a coherent earlier preference.
Sentence-T5 item embeddings are partitioned into 16 interest clusters using MiniBatch K-means~\cite{kmeans}, with the held-out target determining the target-interest cluster.
An $s$-touch instance ($s\in\{1,2\}$) ends with exactly $s$ consecutive interactions from this cluster, with at least three of the four immediately preceding interactions belonging to another common cluster.
Histories shorter than $s+4$ are excluded.
Ground-truth targets are used only for constructing method-independent evaluation subsets, and the one- and two-touch cohorts comprise distinct instances.
The resulting cohorts contain 602 and 105 instances on Science, and 350 and 54 on Yelp, for one- and two-touch conditions, respectively.

Table~\ref{tab:emerging_interest} shows that TCBSA consistently improves the frozen TIGER backbone in both Recall@10 and NDCG@10 across all four emerging-interest subsets.
Relative NDCG@10 gains reach 16.29\% and 31.95\% on Science and 76.31\% and 103.29\% on Yelp for one- and two-touch histories, respectively, with the larger Yelp gains partly reflecting TIGER's weaker baseline performance.
Compared with training-time enhancement and post-hoc fusion, TCBSA also outperforms Human-Prior in both metrics and achieves higher NDCG@10 than both fusion baselines across all four subsets, while matching or exceeding their Recall@10.
These results further support the benefit of intervening before beam pruning, particularly when the target interest is represented by sparse recent interactions.

Performance relative to standalone SASRec reveals a more nuanced pattern.
While TCBSA leads both metrics on the Science subsets, SASRec remains stronger on Yelp's one-touch subset; with two touches, TCBSA achieves the highest NDCG@10 and ties for the best Recall@10.
The contrast is consistent with additional recent interactions strengthening the target-interest signal, although the distinct cohorts and limited two-touch samples preclude a within-user interpretation.

\begin{figure*}[t]
    \centering
    \includegraphics[width=\textwidth]{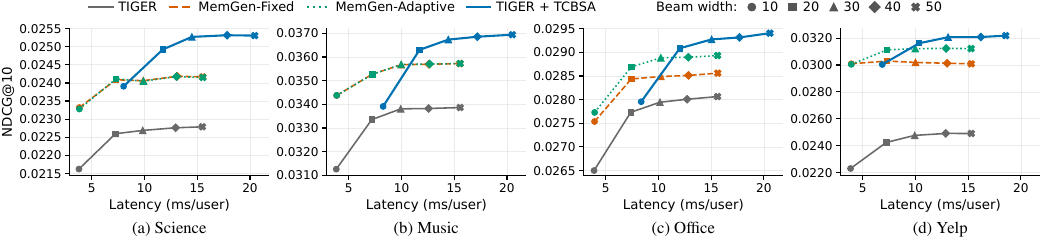}
    \caption{
    Accuracy--latency trade-offs across beam widths on four datasets.
    NDCG@10 is plotted against mean synchronized inference time per user (ms/user).
    Line styles distinguish methods, while marker shapes denote $B\in\{10,20,30,40,50\}$.
    }
    \label{fig:beam_tradeoff}
\end{figure*}

\subsection{Analysis of Decoding Behavior (RQ3)}
\label{sec:search_mechanism}

To understand how test-time collaborative adaptation reshapes semantic decoding, we trace target-prefix survival and recovery under native TIGER and TCBSA on Yelp with beam width $B=20$.

\paragraph{Target-Prefix Preservation and Recovery.}
Figure~\ref{fig:search_mechanism}(a) shows that native TIGER progressively loses target prefixes, whereas TCBSA maintains higher survival rates at every semantic depth. At the final depth, survival increases from 6.72\% to 8.37\% (+24.44\%), closely matching the item-level HR@20 of 8.35\% after valid-item completion. The widening gap across depths suggests that pre-pruning adaptation helps prevent the accumulation of irreversible pruning errors.

To examine whether local rescues persist through subsequent decoding, Figure~\ref{fig:search_mechanism}(b) traces their progression from same-step rescue to complete-path preservation and final Hit@10. Among 28,385 instances where native TIGER first prunes the target prefix, TCBSA rescues 2,163 at the same depth (7.62\%). Of these rescues, 985 retain the complete target path (45.54\%), and 632 of the preserved paths enter the final Top-10 list (64.16\%). These transitions show that test-time score corrections can propagate beyond the initial intervention and recover correct recommendations.

Figure~\ref{fig:search_mechanism}(c) further examines the native token ranks of all same-step rescues, irrespective of subsequent path preservation. Ranks are computed under ground-truth preceding prefixes. The median rank is 22, with 56.68\% of rescued tokens falling outside the native Top-20, indicating that collaborative guidance can promote semantic extensions that the generator initially ranks poorly. Although TCBSA also displaces 12.68\% of target prefixes surviving the first native decoding depth, the total number of Hit@10 instances increases from 1,393 to 1,738, yielding a net gain of 345 correct recommendations.

\paragraph{Recommendations Beyond the Discriminative Top-$R$ Set.}
TCBSA constructs collaborative potentials from retrieved items without restricting generative outputs to the discriminative Top-$R$ set. Figure~\ref{fig:outside_hits} measures the proportion of correct TCBSA Top-10 predictions whose targets lie outside the corresponding SASRec Top-$R$ set, using fixed checkpoints and identical non-$R$ settings. At $R=20$, these outside-set hits account for 21.49\%--38.53\% of all correct recommendations across the four datasets. At the default $R=500$, the proportions are 5.16\% (123/2,383) on Science, 4.72\% (532/11,260) on Office, 0.89\% (34/3,823) on Music, and 0.86\% (15/1,738) on Yelp. Outside-set hits persist even at $R=2000$, confirming that the retrieved set does not impose a hard constraint on generative outputs. Whether shared-prefix corrections directly contribute to these predictions is examined through targeted interventions in Section~\ref{sec:ablation_attribution}.

\subsection{Ablation and Attribution Analysis (RQ4)}
\label{sec:ablation_attribution}

\input{tables/ablation_attribution.tex}

We assess the roles of confidence tempering, semantic-edge aggregation, and shared-prefix corrections through the ablations and targeted interventions in Table~\ref{tab:ablation_attribution}.

\paragraph{Confidence Gate and Aggregation.}
Table~\ref{tab:ablation_attribution}(a) presents Science and Yelp results, while Table~\ref{tab:appendix_gate_aggregation} covers all four datasets.
Removing the gate and applying the maximum intervention strength uniformly reduces NDCG@10 by 1.88\%--9.08\%.
This supports confidence-tempered adaptation, whose score-perturbation behavior is analyzed in Appendix~\ref{app:theory_proofs}.
Across top-$M$ mean, mean, max, and log-sum-exp aggregation, the largest absolute relative deviations from the selected configurations range from 0.83\% to 1.89\%.
No aggregation rule consistently performs best across datasets, suggesting limited sensitivity to the aggregation choice.
Detailed results are provided in Appendix~\ref{app:additional_experiments}.

\paragraph{Shared-Prefix Attribution.}
To examine whether shared-prefix corrections contribute to recommendations beyond the discriminative Top-$R$ set, we replay TCBSA on its original outside-set Top-10 hits at $R=500$, removing positive corrections along each ground-truth target path.
For these outside-set targets, positive corrections arise from semantic edges shared with retrieved items.
As controls, we remove approximately strength-matched off-target corrections at the same semantic depths, averaging hit-loss rates over five random trials.
Target-path removal eliminates 159 of 532 (29.89\%) original outside-set hits on Office and 35 of 123 (28.46\%) on Science, compared with negligible losses of 0.04\% and 0.00\% under matched controls.
Moreover, all newly gained outside-set hits on both datasets (34/34 on Office and 7/7 on Science) disappear after target-path removal.
These interventions provide direct evidence that shared-prefix corrections contribute to correct recommendations beyond the retrieved set, rather than merely co-occurring with them.
Ground-truth targets are used only for offline diagnosis, and a lost Top-10 hit may reflect reranking rather than beam pruning.

Additional retrieval-width and hyperparameter analyses in Appendix~\ref{app:additional_experiments} show diminishing returns beyond $R=500$ and examine sensitivity to adaptation strength.

\subsection{Accuracy--Latency Trade-offs (RQ5)}
\label{sec:accuracy_latency}

Increasing beam width can mitigate premature pruning, but at additional decoding cost. To examine whether test-time collaborative adaptation remains beneficial under larger search budgets, Figure~\ref{fig:beam_tradeoff} compares NDCG@10 and inference latency across $B\in\{10,20,30,40,50\}$ with fixed checkpoints and non-beam hyperparameters. Latency is measured as mean synchronized inference time per user over batches of 32, including model inference and method-specific discriminative scoring, fusion, and collaborative adaptation, while excluding loading, tokenization, metric computation, and file I/O.

TCBSA outperforms native TIGER in all 20 dataset--beam configurations, with NDCG@10 gains of 4.67\%--34.68\%. Even at $B=50$, TCBSA maintains NDCG@10 improvements of 11.06\%, 9.10\%, 4.78\%, and 29.21\% on Science, Music, Office, and Yelp, respectively. TCBSA also surpasses both post-hoc fusion baselines at every $B\geq20$, although it slightly trails the stronger fusion baseline on Music and Yelp at $B=10$. Thus, wider native search does not eliminate the benefits of test-time collaborative adaptation.

At the same beam width, these gains come with additional inference cost: at the default $B=20$, TCBSA incurs $1.43\times$--$1.63\times$ the latency of TIGER. However, TCBSA at $B=20$ achieves 3.64\%--26.94\% higher NDCG@10 than TIGER at $B=50$, while reducing latency by 22.95\%--32.27\%. It also outperforms both $B=50$ fusion baselines in accuracy at lower latency. Moreover, increasing TCBSA's beam width from 30 to 50 yields only 0.14\%--0.57\% additional NDCG@10 at 36.78\%--42.09\% higher latency. Together, these observations show that test-time collaborative adaptation can use limited inference budgets more effectively than simply widening the native beam.

%% file: tables/main_table.tex
\begin{table*}[!htbp]
\centering
\caption{
Full-catalog Recall and NDCG at $K\in\{5,10\}$.
Best and second-best results are bold and underlined.
$^{\dagger}$ and $^{*}$ indicate significant gains over the corresponding backbone and the strongest baseline under paired bootstrap ($p<0.05$), respectively.
}
\label{tab:overall_performance}
\scriptsize
\setlength{\tabcolsep}{1.4pt}
\renewcommand{\arraystretch}{1.08}
\begin{tabular*}{\textwidth}{@{\extracolsep{\fill}}l*{16}{r}@{}}
\toprule
\rowcolor{gray!8}
\textbf{Method}
& \multicolumn{4}{c}{\textbf{Science}}
& \multicolumn{4}{c}{\textbf{Music}}
& \multicolumn{4}{c}{\textbf{Office}}
& \multicolumn{4}{c}{\textbf{Yelp}} \\
\cmidrule(lr){2-5}
\cmidrule(lr){6-9}
\cmidrule(lr){10-13}
\cmidrule(lr){14-17}
& \textbf{R@5} & \textbf{N@5} & \textbf{R@10} & \textbf{N@10}
& \textbf{R@5} & \textbf{N@5} & \textbf{R@10} & \textbf{N@10}
& \textbf{R@5} & \textbf{N@5} & \textbf{R@10} & \textbf{N@10}
& \textbf{R@5} & \textbf{N@5} & \textbf{R@10} & \textbf{N@10} \\
\midrule
GRU4Rec
& 0.0200 & 0.0128 & 0.0335 & 0.0171
& 0.0357 & 0.0232 & 0.0560 & 0.0297
& 0.0225 & 0.0151 & 0.0331 & 0.0185
& 0.0233 & 0.0151 & 0.0397 & 0.0204 \\
SASRec
& 0.0238 & 0.0152 & 0.0380 & 0.0198
& 0.0386 & 0.0255 & 0.0584 & 0.0318
& 0.0264 & 0.0177 & 0.0385 & 0.0215
& 0.0321 & 0.0221 & 0.0508 & 0.0281 \\
BERT4Rec
& 0.0163 & 0.0105 & 0.0269 & 0.0139
& 0.0316 & 0.0202 & 0.0491 & 0.0258
& 0.0192 & 0.0125 & 0.0286 & 0.0155
& 0.0185 & 0.0114 & 0.0325 & 0.0159 \\
\baselineSep
TIGER
& 0.0270 & 0.0176 & 0.0426 & 0.0226
& 0.0404 & 0.0268 & 0.0606 & 0.0334
& 0.0334 & 0.0232 & 0.0474 & 0.0277
& 0.0294 & 0.0190 & 0.0458 & 0.0242 \\
HSTU
& 0.0277 & 0.0156 & 0.0435 & 0.0207
& \underline{0.0438} & 0.0274 & 0.0652 & 0.0343
& 0.0267 & 0.0176 & 0.0393 & 0.0217
& 0.0406 & 0.0297 & 0.0593 & 0.0356 \\
MiniOneRec
& 0.0296 & \underline{0.0200} & 0.0430 & 0.0244
& 0.0393 & 0.0270 & 0.0568 & 0.0326
& 0.0315 & 0.0226 & 0.0424 & 0.0262
& \underline{0.0434} & \underline{0.0309} & \underline{0.0607} & \underline{0.0365} \\
\baselineSep
Human-Prior
& 0.0281 & 0.0185 & 0.0438 & 0.0236
& 0.0360 & 0.0234 & 0.0551 & 0.0295
& 0.0143 & 0.0092 & 0.0231 & 0.0121
& 0.0282 & 0.0196 & 0.0421 & 0.0241 \\
APAO-Pairwise
& 0.0276 & 0.0181 & 0.0442 & 0.0233
& 0.0401 & 0.0269 & 0.0611 & 0.0337
& 0.0338 & 0.0233 & 0.0489 & 0.0281
& 0.0302 & 0.0195 & 0.0490 & 0.0256 \\
\baselineSep
MemGen-Fixed
& 0.0292 & 0.0191 & 0.0446 & 0.0241
& 0.0427 & 0.0285 & 0.0637 & 0.0353
& 0.0342 & 0.0238 & 0.0486 & 0.0284
& 0.0349 & 0.0237 & 0.0555 & 0.0303 \\
MemGen-Adaptive
& 0.0290 & 0.0191 & 0.0447 & 0.0241
& 0.0427 & 0.0285 & 0.0637 & 0.0353
& 0.0346 & \underline{0.0241} & 0.0490 & \underline{0.0287}
& 0.0366 & 0.0248 & 0.0564 & 0.0311 \\
\midrule
\rowcolor{blue!6}
\textbf{TIGER + TCBSA(SASRec)}
& \underline{0.0297}$^{\dagger}$ & 0.0195$^{\dagger}$ & \textbf{0.0467}$^{\dagger *}$ & \underline{0.0249}$^{\dagger}$
& \textbf{0.0447}$^{\dagger}$ & \textbf{0.0295}$^{\dagger *}$ & \textbf{0.0658}$^{\dagger}$ & \textbf{0.0363}$^{\dagger *}$
& \textbf{0.0352}$^{\dagger *}$ & \textbf{0.0244}$^{\dagger *}$ & \textbf{0.0499}$^{\dagger *}$ & \textbf{0.0291}$^{\dagger *}$
& 0.0380$^{\dagger}$ & 0.0255$^{\dagger}$ & 0.0571$^{\dagger}$ & 0.0316$^{\dagger}$ \\
\rowcolor{blue!6}
\textbf{TIGER + TCBSA(GRU4Rec)}
& 0.0295$^{\dagger}$ & 0.0193$^{\dagger}$ & \underline{0.0465}$^{\dagger *}$ & 0.0248$^{\dagger}$
& 0.0427$^{\dagger}$ & 0.0282$^{\dagger}$ & \underline{0.0653}$^{\dagger}$ & \underline{0.0355}$^{\dagger}$
& \underline{0.0347}$^{\dagger}$ & 0.0239$^{\dagger}$ & \underline{0.0491}$^{\dagger}$ & 0.0286$^{\dagger}$
& 0.0329$^{\dagger}$ & 0.0216$^{\dagger}$ & 0.0521$^{\dagger}$ & 0.0278$^{\dagger}$ \\
\rowcolor{blue!6}
\textbf{MiniOneRec + TCBSA(SASRec)}
& \textbf{0.0314}$^{\dagger *}$ & \textbf{0.0214}$^{\dagger *}$ & 0.0458$^{\dagger}$ & \textbf{0.0260}$^{\dagger *}$
& 0.0426$^{\dagger}$ & \underline{0.0291}$^{\dagger}$ & 0.0610$^{\dagger}$ & 0.0350$^{\dagger}$
& 0.0323$^{\dagger}$ & 0.0231$^{\dagger}$ & 0.0436$^{\dagger}$ & 0.0268$^{\dagger}$
& \textbf{0.0481}$^{\dagger *}$ & \textbf{0.0335}$^{\dagger *}$ & \textbf{0.0682}$^{\dagger *}$ & \textbf{0.0400}$^{\dagger *}$ \\
\bottomrule
\end{tabular*}
\end{table*}

%% file: tables/emerging_interest.tex
\begin{table}[t]
\centering

\caption{
Recall@10 and NDCG@10 on one- and two-touch emerging-interest subsets of Science and Yelp.
TCBSA uses TIGER with SASRec-based collaborative adaptation.
Best and second-best results are bold and underlined.
}

\label{tab:emerging_interest}
\small
\setlength{\tabcolsep}{2.5pt}
\renewcommand{\arraystretch}{1.06}
\begin{tabular*}{\columnwidth}{@{\extracolsep{\fill}}lrrrr@{}}
\toprule
\rowcolor{gray!8}
\textbf{Method}
& \multicolumn{2}{c}{\textbf{1-touch}}
& \multicolumn{2}{c}{\textbf{2-touch}} \\
\cmidrule(lr){2-3}
\cmidrule(lr){4-5}
\rowcolor{gray!8}
& \textbf{R@10}
& \textbf{N@10}
& \textbf{R@10}
& \textbf{N@10} \\
\midrule

\rowcolor{gray!5}
\multicolumn{5}{l}{\textit{Science}} \\
TIGER
& 0.0698 & 0.0398
& 0.1143 & 0.0722 \\
SASRec
& 0.0648 & 0.0392
& \underline{0.1429} & 0.0884 \\
Human-Prior
& 0.0664 & 0.0369
& 0.1143 & \underline{0.0904} \\
MemGen-Fixed
& \underline{0.0748} & 0.0438
& 0.1333 & 0.0870 \\
MemGen-Adaptive
& \underline{0.0748} & \underline{0.0440}
& 0.1333 & 0.0873 \\
\rowcolor{blue!6}
\textbf{TCBSA}
& \textbf{0.0781} & \textbf{0.0462}
& \textbf{0.1619} & \textbf{0.0953} \\

\addlinespace[2pt]
\rowcolor{gray!5}
\multicolumn{5}{l}{\textit{Yelp}} \\
TIGER
& 0.0600 & 0.0284
& 0.0370 & 0.0302 \\
SASRec
& \textbf{0.0886} & \textbf{0.0525}
& \textbf{0.0741} & 0.0506 \\
Human-Prior
& 0.0629 & 0.0292
& 0.0370 & 0.0234 \\
MemGen-Fixed
& \underline{0.0829} & 0.0501
& \textbf{0.0741} & 0.0535 \\
MemGen-Adaptive
& 0.0800 & 0.0496
& \textbf{0.0741} & \underline{0.0546} \\
\rowcolor{blue!6}
\textbf{TCBSA}
& \underline{0.0829} & \underline{0.0502}
& \textbf{0.0741} & \textbf{0.0614} \\

\bottomrule
\end{tabular*}
\end{table}

%% file: tables/ablation_attribution.tex
\begin{table}[t]
\centering
\caption{
TCBSA design ablations and shared-prefix attribution.
(a) NDCG@10 under confidence-gate and aggregation variants.
(b) Loss of correct Top-10 hits outside the SASRec Top-$R$ set
($R=500$) after removing target-path corrections, compared
with matched off-target controls.
Relative changes are computed from unrounded values.
}
\label{tab:ablation_attribution}

\small
\setlength{\tabcolsep}{2pt}
\renewcommand{\arraystretch}{1.2}

\begin{tabular*}{\columnwidth}{
@{\extracolsep{\fill}}lccccc@{}
}
\toprule
\multicolumn{6}{@{}l}{
\textit{(a) Confidence gate and aggregation}
} \\
\textbf{Dataset}
& \textbf{Full}
& \textbf{No gate}
& \textbf{Drop}
& \textbf{Agg. range}
& \textbf{Max dev.} \\
\midrule
Science
& 0.0249
& 0.0233
& 6.65\%
& 0.0247--0.0251
& 0.83\% \\
Yelp
& 0.0316
& 0.0310
& 1.88\%
& 0.0310--0.0318
& 1.89\% \\
\end{tabular*}

\par\vspace{0.2em}

\begin{tabular*}{\columnwidth}{
@{\extracolsep{\fill}}lcccc@{}
}
\midrule
\multicolumn{5}{@{}l}{
\textit{(b) Outside-retrieval target-path attribution}
} \\
\textbf{Dataset}
& \textbf{Hits}
& \textbf{Target drop}
& \textbf{Control drop}
& \textbf{New lost} \\
\midrule
Office
& 532
& 29.89\%
& 0.04\%
& 34/34 \\
Science
& 123
& 28.46\%
& 0.00\%
& 7/7 \\
\bottomrule
\end{tabular*}

\end{table}

%% file: section/conclusion.tex
\section{Conclusion}
TCBSA addresses irreversible prefix pruning in semantic-ID-based generative recommendation by projecting discriminative item-score margins onto shared semantic edges and applying confidence-tempered corrections to beam scores at test time, without updating either model.
We establish a sufficient condition for target-prefix survival and bound score perturbations under low discriminative confidence.
Across four full-catalog datasets, TCBSA variants achieve the best results on all 16 dataset--metric pairs, with significant gains over the strongest baselines in 14 comparisons.
These improvements extend across different generative--discriminative model pairings.
Emerging-interest and decoding analyses show accuracy gains and target-prefix recovery, while targeted interventions demonstrate that shared-prefix corrections contribute to correct recommendations beyond the discriminative Top-$R$ set.
Although TCBSA incurs additional latency at matched beam widths, it achieves higher NDCG@10 and lower latency at $B=20$ than native TIGER at $B=50$.
These findings highlight the value of incorporating collaborative scores before beam pruning, rather than only combining final rankings.

%% file: section/appendix.tex
\section{Theoretical Analysis and Complexity}
\label{app:theory_proofs}

We analyze TCBSA's cumulative test-time score adaptation, establish a
sufficient condition for target-prefix survival under global beam
competition, and characterize its inference complexity.

\subsection{Cumulative Test-Time Score Decomposition}
\label{app:score_decomposition}

For a partial semantic path $h=(h_1,\ldots,h_\ell)$, define its
accumulated collaborative potential as
$C_{\mathbf{x}_u}(h)=\sum_{j=1}^{\ell}
A_{\mathbf{x}_u}(h_{<j},h_j)$.
Since the history-specific gate $\eta(\mathbf{x}_u)$ is fixed during
decoding, recursively applying the pre-pruning update in
Equation~\eqref{eq:guided_path_score} gives
\begin{equation}
\begin{aligned}
\widetilde{S}_u(h)
&=S_0(h\mid\mathbf{x}_u)+\eta(\mathbf{x}_u)C_{\mathbf{x}_u}(h),\\
\Delta_u(h)
&=\eta(\mathbf{x}_u)C_{\mathbf{x}_u}(h),
\end{aligned}
\label{eq:cumulative_adapted_score}
\end{equation}
where $S_0(h\mid\mathbf{x}_u)=\sum_{j=1}^{\ell}
\log\pi_0(h_j\mid h_{<j},\mathbf{x}_u)$ is the native cumulative
score. Thus, test-time adaptation modifies path scores without
changing the generator's token probabilities.

\subsection{Target-Prefix Survival under Global Beam Competition}
\label{app:target_prefix_survival}

\begin{theorem}[Target-Prefix Survival]
\label{thm:target_prefix_survival}
At decoding depth $\ell$, let $\mathcal{F}_{u,\ell}$ be the global
candidate frontier, and suppose the target prefix $y_{1:\ell}$
belongs to $\mathcal{F}_{u,\ell}$.
Let $\mathcal{H}_{u,\ell}\subseteq
\mathcal{F}_{u,\ell}\setminus\{y_{1:\ell}\}$ contain competing
hypotheses, with at most $B-1$ competitors outside
$\mathcal{H}_{u,\ell}$. For each $h\in\mathcal{H}_{u,\ell}$, define
\begin{equation}
\begin{aligned}
\Delta_C(h)
&=C_{\mathbf{x}_u}(y_{1:\ell})-C_{\mathbf{x}_u}(h),\\
\Delta_S(h)
&=S_0(h\mid\mathbf{x}_u)-S_0(y_{1:\ell}\mid\mathbf{x}_u).
\end{aligned}
\label{eq:beam_competition_margins}
\end{equation}
If every $h\in\mathcal{H}_{u,\ell}$ satisfies
\begin{equation}
\eta(\mathbf{x}_u)\Delta_C(h)>[\Delta_S(h)]_+,
\label{eq:beam_survival_condition}
\end{equation}
then $y_{1:\ell}$ is retained in the global top-$B$ beam.
\end{theorem}

\begin{proof}
For any $h\in\mathcal{H}_{u,\ell}$, the adapted score difference is
\begin{equation}
\begin{aligned}
\widetilde{S}_u(y_{1:\ell})-\widetilde{S}_u(h)
&=-\Delta_S(h)+\eta(\mathbf{x}_u)\Delta_C(h)\\
&> -\Delta_S(h)+[\Delta_S(h)]_+\\
&\geq 0.
\end{aligned}
\end{equation}
Hence, every hypothesis in $\mathcal{H}_{u,\ell}$ scores strictly
below the target prefix. At most $B-1$ remaining competitors can
score at least as high as the target, so its global rank is at most
$B$, and it survives beam pruning.
\end{proof}

Applying the theorem successively at every depth, starting from the
empty prefix, also preserves the complete target path whenever the
stated condition holds at each depth.

\paragraph{Low-confidence score perturbation.}
The gate additionally limits changes to pairwise path-score margins.
For two length-$\ell$ paths $h,h'$, let
\begin{equation}
\begin{aligned}
\Omega_{\mathbf{x}_u}(h,h')
=\sum_{j=1}^{\ell}\big|&A_{\mathbf{x}_u}(h_{<j},h_j)\\
&-A_{\mathbf{x}_u}(h'_{<j},h'_j)\big|.
\end{aligned}
\end{equation}
When $m(\mathbf{x}_u)\leq\tau$, the sigmoid gate satisfies
$\eta(\mathbf{x}_u)\leq
\lambda e^{-k(\tau-m(\mathbf{x}_u))}$. Therefore,
\begin{equation}
\begin{aligned}
|\Delta_u(h)-\Delta_u(h')|
&\leq\lambda e^{-k(\tau-m(\mathbf{x}_u))}\\
&\quad\cdot\Omega_{\mathbf{x}_u}(h,h').
\end{aligned}
\label{eq:low_confidence_stability}
\end{equation}
This follows from the triangle inequality and
$\sigma(a)\leq e^a$. Thus, the confidence gate bounds changes to native pairwise score differences. If the absolute native score gap exceeds this bound, the corresponding ordering is preserved.

\subsection{Inference Complexity}
\label{app:inference_complexity}

\paragraph{Offline preprocessing.}
The item-to-path lookup stores $L$ prefix--token edges for each item, requiring $O(|\mathcal{I}|L)$ preprocessing time and storage.
It is constructed once and reused across users.

\paragraph{Test-time collaborative potential construction.}
Let $\mathcal{T}_D(|\mathcal{I}|,R)$ denote the cost of full-catalog scoring, confidence computation, and Top-$R$ selection by the frozen discriminative model.
Mapping the $R$ selected items onto their $L$ semantic edges produces at most $RL$ margin contributions.
Max aggregation requires $O(RL)$ time; top-$M$ mean aggregation requires $O(RL\log M)$ time with size-$M$ heaps, reducing to $O(RL)$ for fixed $M$.
The resulting sparse collaborative potentials require at most $O(RL)$ storage.

\paragraph{Adapted beam decoding.}
Let $\mathcal{T}_G(B,L)$ denote the computational cost of native beam search, including candidate expansion.
TCBSA adds one potential lookup and score addition per evaluated expansion, preserving the native decoder's asymptotic scaling under constant-time edge lookup.
For fixed $M$, total online time is
\begin{equation}
O\!\left(
\mathcal{T}_D(|\mathcal{I}|,R)
+\mathcal{T}_G(B,L)+RL
\right),
\label{eq:inference_complexity}
\end{equation}
with $O(RL)$ additional storage for sparse potentials, excluding buffers used for full-catalog scoring and native decoding.
This asymptotic bound does not imply negligible latency overhead in practice.

\section{Additional Experimental Details}
\label{app:implementation_details}

\subsection{Dataset Construction and Statistics}
\label{app:dataset_details}

For the three Amazon domains, we use the official 5-core chronological benchmark sequences from Amazon Reviews 2023~\cite{AmazonReviews2023}, without additional rating thresholds or deduplication. Item text for semantic-ID construction combines titles, features, categories, and descriptions.
For Yelp, we use version 2 of the Kaggle release~\cite{YelpKaggleV2}, retaining review--business interactions from the configured 2019 window. We treat all star ratings as implicit interactions, sort reviews chronologically, and iteratively enforce the 5-core constraint. Business metadata supplies the item text.

Table~\ref{tab:dataset_statistics} reports statistics after preprocessing and before leave-one-out splitting. Each user's last interaction is held out for testing and the second-to-last for validation; the history-length limit of 20 applies only to model inputs and does not affect these statistics.

\begin{table}[h]
\centering
\caption{Dataset statistics after preprocessing and before leave-one-out splitting. Avg.\ len.\ is the mean interactions per user; sparsity is the fraction of unobserved user--item pairs.}
\label{tab:dataset_statistics}
\resizebox{\columnwidth}{!}{%
\begin{tabular}{lrrrrr}
\toprule
Dataset & Users & Items & Interactions & Avg.\ len. & Sparsity \\
\midrule
Science & 50{,}985 & 25{,}848 & 412{,}947 & 8.10 & 99.97\% \\
Music   & 58{,}062 & 24{,}701 & 520{,}366 & 8.96 & 99.96\% \\
Office  & 225{,}726 & 77{,}929 & 1{,}826{,}077 & 8.09 & 99.99\% \\
Yelp    & 30{,}431 & 20{,}033 & 316{,}354 & 10.40 & 99.95\% \\
\bottomrule
\end{tabular}%
}
\end{table}

\subsection{Final Configurations and Computing Environment}
\label{app:hyperparameters}

TCBSA hyperparameters are selected on validation NDCG@10 using TIGER--SASRec. Retrieval width $R=500$ is shared across datasets, as selected by macro-averaged validation performance. The selected (aggregation, $\lambda$, $k$, $\tau$) configurations are (top-8 mean, 1.2, 13, 0.125) for Science and Music, (top-8 mean, 0.8, 13, 0.100) for Office, and (max, 1.0, 13, 0.075) for Yelp. Each dataset's settings are reused unchanged for TIGER--GRU4Rec and MiniOneRec--SASRec, without pairing-specific tuning. All experiments use NVIDIA RTX 5880 Ada Generation accelerators, and inference-efficiency comparisons use the same accelerator model.

For MemGen-Fixed, the validation-selected mixing weights are 0.4, 0.5, 0.4, and 0.7 for Science, Music, Office, and Yelp, respectively.
The corresponding MemGen-Adaptive gate parameters
$(k_{\mathrm e},\tau_{\mathrm e})$ are $(1,0.4)$,
$(1,0.0)$, $(9,0.1)$, and $(5,0.0)$.
All fusion parameters are selected on validation data.

\section{Additional Experimental Analyses}
\label{app:additional_experiments}

\subsection{Shared-Prefix Attribution}
\label{app:shared_prefix_attribution}

We extend the targeted attribution analysis in Table~\ref{tab:ablation_attribution} to all four datasets at the default retrieval width $R=500$.
For each original correct TCBSA Top-10 hit outside the SASRec Top-$R$ set, we replay frozen decoding after removing positive retrieval-derived corrections along the ground-truth target path.
As a control, we remove approximately strength-matched off-target corrections at the same semantic depths and baseline-active parents, averaging hit losses over five random trials.
The ground-truth target is used only for this offline diagnostic.

\begin{table}[h]
\centering
\caption{
Shared-prefix attribution on correct TCBSA Top-10 hits outside
the SASRec Top-$R$ set ($R=500$).
Target and control drop denote the fraction of original hits
lost after target-path removal and five-trial matched off-target
removal, respectively.
New lost counts newly gained outside-set hits over TIGER
eliminated by target-path removal.
}
\label{tab:full_prefix_attribution}

\small
\setlength{\tabcolsep}{3pt}
\renewcommand{\arraystretch}{1.15}

\begin{tabular*}{\columnwidth}{
@{\extracolsep{\fill}}lrrrr@{}
}
\toprule
\textbf{Dataset}
& \textbf{Hits}
& \textbf{Target drop}
& \textbf{Control drop}
& \textbf{New lost} \\
\midrule
Office  & 532 & 29.89\% & 0.04\% & 34/34 \\
Science & 123 & 28.46\% & 0.00\% & 7/7   \\
Music   & 34  & 32.35\% & 0.00\% & 1/1   \\
Yelp    & 15  & 80.00\% & 0.00\% & 5/5   \\
\bottomrule
\end{tabular*}
\end{table}

Table~\ref{tab:full_prefix_attribution} shows that removing target-path corrections eliminates 11/34 original outside-set hits on Music and 12/15 on Yelp, with no losses under their matched controls.
Across all four datasets, all 47 newly gained outside-set hits disappear after target-path removal.
Of the 159, 35, 11, and 12 lost original hits on Office, Science, Music, and Yelp, respectively, only 44, 9, 6, and 9 involve observed target-path beam loss; the remaining 115, 26, 5, and 3 disappear without observed beam pruning.
Thus, these targeted interventions establish dependence of some predictions on shared-prefix corrections, but do not imply that every lost hit results from beam pruning or that recall improves over the full outside-retrieval cohort.

The matched controls are available for all outside-set hits.
Their mean strength-matching errors are 4.81\%, 5.57\%,
4.30\%, and 7.18\% on Office, Science, Music, and Yelp,
respectively, with 99.42\%, 99.90\%, 99.07\%, and
91.03\% of control edges active during baseline decoding.
Positive target-path support is observed in 505/532,
118/123, 33/34, and 15/15 cases, respectively.
These diagnostics characterize the quality and coverage
of the interventions without assuming identical beam
effects between target and control removals.

\begin{figure*}[t]
\centering
\includegraphics[width=0.9\textwidth]{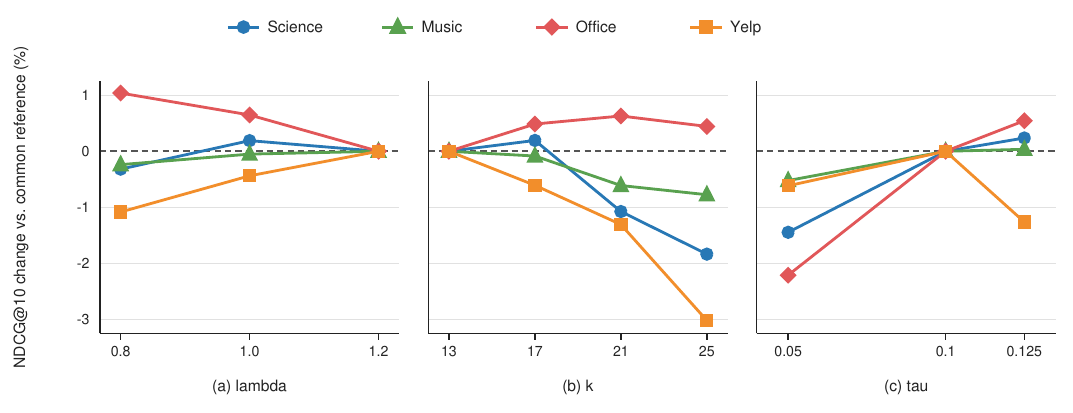}
\caption{Cross-dataset hyperparameter sensitivity on complete validation sets. Panels (a)--(c) vary the adaptation scale $\lambda$, gate slope $k$, and confidence threshold $\tau$, respectively. Curves show relative NDCG@10 changes from the common reference $(\lambda,k,\tau)=(1.2,13,0.1)$; the other two parameters remain fixed. Only values evaluated on all four datasets are shown.}
\label{fig:appendix_parameter_sensitivity}
\Description{Three panels show relative NDCG@10 changes as lambda, k, and tau vary on Science, Music, Office, and Yelp.}
\end{figure*}
\subsection{Design and Parameter Sensitivity}
\label{app:parameter_sensitivity}

\paragraph{Confidence gate and aggregation.}
We evaluate confidence tempering and semantic-edge aggregation on the complete test sets of all four datasets, modifying one component at a time while retaining each dataset's validation-selected configuration.
In the no-gate variant, the intervention strength is fixed at $\lambda$.
For aggregation, we compare top-$M$ mean ($M=8$), mean, max, and log-sum-exp, keeping the remaining settings unchanged.

\begin{table}[h]
\centering
\caption{
Test-set ablations of the confidence gate and semantic-edge aggregation.
Values are relative NDCG@10 changes (\%) from the validation-selected TCBSA configuration, computed using unrounded scores.
Bold marks the selected aggregation rule.
Top-8 denotes top-8 mean; LSE denotes log-sum-exp.
}
\label{tab:appendix_gate_aggregation}

\small
\setlength{\tabcolsep}{2pt}
\renewcommand{\arraystretch}{1.18}

\begin{tabular*}{\columnwidth}{
@{\extracolsep{\fill}}lcrrrr@{}
}
\toprule
& \multicolumn{1}{c}{\textbf{Gate}}
& \multicolumn{4}{c}{\textbf{Aggregation}} \\
\cmidrule(lr){2-2}
\cmidrule(l){3-6}
\textbf{Dataset}
& \textbf{No gate}
& \textbf{Top-8}
& \textbf{Mean}
& \textbf{Max}
& \textbf{LSE} \\
\midrule
Science
& $-6.65$
& $\mathbf{0.00}$
& $-0.49$
& $+0.56$
& $-0.83$ \\
Music
& $-6.01$
& $\mathbf{0.00}$
& $-0.22$
& $+1.19$
& $-0.28$ \\
Office
& $-9.08$
& $\mathbf{0.00}$
& $+0.14$
& $-0.22$
& $-1.45$ \\
Yelp
& $-1.88$
& $+0.48$
& $-0.03$
& $\mathbf{0.00}$
& $-1.89$ \\
\bottomrule
\end{tabular*}
\end{table}

Table~\ref{tab:appendix_gate_aggregation} shows that removing the gate reduces NDCG@10 by 1.88\%--9.08\% across all four datasets.
Across aggregation rules, the largest absolute relative deviations from the validation-selected configurations are 0.83\%, 1.19\%, 1.45\%, and 1.89\% on Science, Music, Office, and Yelp, respectively.
Some alternative aggregators achieve slightly higher test NDCG@10, but the validation-selected configurations remain unchanged.

\paragraph{Confidence-parameter sensitivity.}
We conduct one-factor sensitivity analysis on the complete validation sets using $(\lambda,k,\tau)=(1.2,13,0.1)$ as a common reference.
We vary $\lambda\in\{0.8,1.0,1.2\}$, $k\in\{13,17,21,25\}$, and $\tau\in\{0.05,0.10,0.125\}$, keeping the other two parameters fixed.
The reference differs from the dataset-specific final configurations in Appendix~\ref{app:implementation_details}.

Figure~\ref{fig:appendix_parameter_sensitivity} shows limited sensitivity over the shared grids.
The largest relative NDCG@10 decreases when varying $\lambda$, $k$, and $\tau$ are 1.08\% (Yelp), 3.02\% (Yelp), and 2.21\% (Office), respectively.
These results describe robustness around the common reference, not dataset-specific retuning.

\subsection{Retrieval-Width Selection}
\label{app:retrieval_width_sensitivity}

We vary $R\in\{20,100,500,1000,2000\}$ on the complete validation sets, keeping all non-$R$ settings at each dataset's final $R=500$ configuration.
Table~\ref{tab:appendix_retrieval_width} reports NDCG@10 by dataset and the macro average.

\begin{table}[h]
\centering
\caption{
Validation NDCG@10 across retrieval widths.
Bold values mark the globally selected $R=500$,
not the per-dataset optimum.
}
\label{tab:appendix_retrieval_width}

\small
\setlength{\tabcolsep}{3pt}
\renewcommand{\arraystretch}{1.15}

\begin{tabular*}{\columnwidth}{
@{\extracolsep{\fill}}lrrrrr@{}
}
\toprule
\textbf{Dataset}
& \textbf{20}
& \textbf{100}
& \textbf{500}
& \textbf{1000}
& \textbf{2000} \\
\midrule
Science
& 0.0301 & 0.0309 & \textbf{0.0316} & 0.0316 & 0.0315 \\
Music
& 0.0399 & 0.0406 & \textbf{0.0403} & 0.0400 & 0.0394 \\
Office
& 0.0353 & 0.0357 & \textbf{0.0358} & 0.0357 & 0.0356 \\
Yelp
& 0.0339 & 0.0364 & \textbf{0.0371} & 0.0371 & 0.0372 \\
\midrule
Macro
& 0.0348 & 0.0359 & \textbf{0.0362} & 0.0361 & 0.0359 \\
\bottomrule
\end{tabular*}
\end{table}

The macro average peaks at $R=500$, although the per-dataset optima differ.
We select the smallest width whose macro-average NDCG@10 lies within 0.5\% of the best grid value and whose largest per-dataset deficit is at most 1.5\%.
At $R=500$, the largest deficit is 0.70\% on Music, computed from unrounded validation results.
No test-set results are used for this selection.